\documentclass[reqno]{amsart}

\newtheorem{thm}{Theorem}[subsection]
\newtheorem{lem}[thm]{Lemma}

\begin{document}

\title[Generalized Differential Transfer Matrices]
{Differential Transfer Matrix Solution of Generalized Eigenvalue Problems}

\author{Sina Khorasani}
\address{School of Electrical Engineering, Sharif University of Technology}
\email{khorasani@sina.sharif.edu}






\begin{abstract}
We report a new analytical method for solution of a wide class of
second-order differential equations with eigenvalues replaced by
arbitrary functions. Such classes of problems occur frequently in 
Quantum Mechanics and Optics. This approach is based on the extension of the
previously reported differential transfer matrix method with modified
basis functions. Applications of the method to boundary value and
initial value problems, as well as several examples are illustrated. 
\end{abstract}

\maketitle

\section{Introduction}

Analytical solution of differential equations provides insight to the
behavior of solutions whenever they exist. Unfortunately, most
practical physical problems are described by governing models which are
normally solved by numerical techniques. Moreover, explicit solutions
even to the simplest differential equations are rare. A wide class of
physical problems are described by second-order differential equations,
where the only known analytical method with explicit solution for this
purpose is the approximate Wentzel\textendash{}Kramers\textendash{}Brillouin (WKB) 
method \cite{ref1,ref2}. An approximate method for solution of linear homogeneous
differential equations with variable coefficients has been reported in \cite{ref2a},
which is based on a transformation into a Volterra or Fredholm integral equation. 
Also, a matrix method has been reported in \cite{ref2b} which is used to transform the 
Schr\"{o}dinger equation into a Shabat–Zakharov system of second order, and 
then the solution is obtained by perturbation technique.
There are also a number of further existing analytical methods for solution of
Ordinary Differential Equations which are categorized and reviewed by Polyanin and 
Zaitsev in \cite{ref2c}.

Recently, we have introduced an analytical method which is capable of
solving linear homogeneous Ordinary Differential Equations (ODEs) with
variable coefficients \cite{ref3}. The method is based on the definition of jump transfer matrices and
taking the differential limit. The approach reduces the
$n$th-order differential equation to a first order system of
$n$ linear differential equations.
The full analytical solution is then found by the perturbation
technique, which may be elegantly expressed in terms of matrix
exponentiation of the integral of a Kernel matrix. The important
feature of this method is that it deals with the evolution of
independent solutions, rather than its derivatives. The exact
mathematical nature of this method has also been rigorously established
\cite{ref3}.

This method emerged as an extension of Differential Transfer Matrix
Method (DTMM), which was originally proposed in the context of optics
and quantum mechanics \cite{ref4,ref5,ref6}. DTMM is based on the modification of the
standard transfer matrix method in optics \cite{ref7} and quantum mechanics \cite{ref8},
and employs exponential basis functions. Through numerous examples,
this method was shown to be simple, exact, and efficient, while
reflecting the basic properties of the physical problem.

However, the initial formulation of DTMM had difficulty when dealing
with singularities. Such singularities arise in the domain of optics
and quantum mechanics at the turning points of wavefunctions, where the
behavior switches between oscillatory and decaying forms \cite{ref2}.
As a result, basic DTMM results suffer from numerical error when
approaching physical singularities. To overcome this difficulty, the
solution may be expressed by using Airy functions \cite{ref2}. Further
applications and extensions of DTMM have been reported by
Mehrany et. al., which include employment of WKB basis \cite{ref9,ref10} for
improvement of numerical accuracy, and numerical implementation of the
method using Airy functions \cite{ref11} for dealing with singularities. Also,
the original DTMM has been used in \cite{ref12} with no essential change as
reported in the main formulation \cite{ref3}.

In this paper, we report a general formulation of DTMM for eigenvalue
problems having second-order ODEs. The generalization is done by
replacing the eigenvalue with an arbitrary function, which we here
refer to as the eigenvalue function. We show that DTMM is still capable
of delivering the analytical solution provided that proper basis
functions are used. We also establish the mathematical validity of this
method by presenting the fundamental theorem of DTMM. The major
improvement of this work in contrast to the original formulation
\cite{ref3} is two-fold. Firstly, the solution to the extended eigenvalue problem
evolves `naturally' out of the relevant basis functions rather than
exponentials. Secondly, envelope functions undergo minimal variations
because most of the solution is encompassed in the extended bases. As a
result, we may obtain a simple approximate solution for problems with
slowly varying eigenvalue functions.

\section{Formulation}

Suppose that a linear homogeneous ordinary differential equation (ODE)
is given as

\begin{equation}
\mathbb{H}_k f\left( {x;k} \right) = 0
\label{eq1}
\end{equation}

\noindent
where $\mathbb{H}_k $ is a second-order linear
operator with eigenvalue $k$ which involves partial
derivatives only with respect to $x$, and
$f\left( {x;k} \right)$ being its eigenfunction. It
is the purpose of this formulation to find solutions of (\ref{eq1})
when the eigenvalue $k$ is replaced by an
eigenvalue function such as $k\left( x \right)$.
Hence, the extended problem reads

\begin{equation}
\mathbb{L}_{k\left( x \right)} f\left[ {x;k\left( x \right)} \right] = 0
\label{eq2}
\end{equation}

\noindent
Here, the new operator $\mathbb{L}_{k\left( x \right)}
$ is obtained from $\mathbb{H}_k $ by
replacing the eigenvalue $k$ with an eigenvalue
function $k\left( x \right)$, and
${\partial  \mathord{\left/
{\vphantom {\partial  {\partial x}}} \right.
\kern-\nulldelimiterspace} {\partial x}}$ with
${d \mathord{\left/
{\vphantom {d {dx}}} \right.
\kern-\nulldelimiterspace} {dx}}$. In general, (\ref{eq1})
admits a general solution having the form

\begin{equation}
f\left( {x;k} \right) = aA\left( {x;k} \right) + bB\left( {x;k}
\right)
\label{eq3}
\end{equation}

\noindent
in which $a$ and $b$ are constants determined by initial or boundary conditions, and
$A\left( {x;k} \right)$ and $B\left({x;k} \right)$ are linearly independent solutions with
non-vanishing Wronskian determinant, i.e.

\begin{equation}
W\left( {x;k} \right) = \left| {\begin{array}{*{20}c}
{A\left( {x;k} \right)} & {B\left( {x;k} \right)}  \\
{A_x \left( {x;k} \right)} & {B_x \left( {x;k} \right)}  \\
\end{array}} \right| \ne 0
\label{eq4}
\end{equation}

\noindent
Clearly, $A_x \left( {x;k} \right) = {{\partial A\left(
{x;k} \right)} \mathord{\left/
{\vphantom {{\partial A\left( {x;k} \right)} {\partial x}}} \right.
\kern-\nulldelimiterspace} {\partial x}}$ and
$B_x \left( {x;k} \right) = {{\partial B\left( {x;k}
\right)} \mathord{\left/
{\vphantom {{\partial B\left( {x;k} \right)} {\partial x}}} \right.
\kern-\nulldelimiterspace} {\partial x}}$.

Now, we seek a solution to (\ref{eq2}) having the extended form to (\ref{eq3})
given by

\begin{equation}
f\left( {x;k} \right) = a\left( x \right)A\left[ {x;k\left( x \right)}
\right] + b\left( x \right)B\left[ {x;k\left( x \right)} \right]
\label{eq5}
\end{equation}

\noindent
where $a\left( x \right)$ and $b\left(x \right)$ remain as unknown functions to be determined.

\subsection{Transfer Matrix of Finite Jumps}

In order to proceed with the formulation, we take a similar approach
to the conventional Differential Transfer Matrix Method with
exponential basis. For this reason, we first need to obtain the
transfer matrix of finite jumps in the eigenvalue
$k$. Suppose that the eigenvalue function is
$k\left( x \right)$ defined as

\begin{equation}
k\left( x \right) = \left\{ {\begin{array}{*{20}c}
{k_1 ,} & {x > X}  \\
{k_2 ,} & {x < X}  \\
\end{array}} \right.
\label{eq6}
\end{equation}

\noindent
with $k_1 $, $k_2 $, and $X$ being constants. Solution to (\ref{eq1})
is readily given by (\ref{eq3}) as

\begin{equation}
\begin{array}{l}
f\left[ {x;k\left( x \right)} \right] = \left\{
{\begin{array}{*{20}c}
{f_1 \left( x \right),} & {x < X}  \\
{f_2 \left( x \right),} & {x > X}  \\
\end{array}} \right. \\
\quad \quad \;\;\;\, = \left\{ {\begin{array}{*{20}c}
{a_1 A\left( {x;k_1 } \right) + b_1 B\left( {x;k_1 } \right),} & {x
< X}  \\
{a_2 A\left( {x;k_2 } \right) + b_2 B\left( {x;k_2 } \right),} & {x
> X}  \\
\end{array}} \right. \\
\end{array}
\label{eq7}
\end{equation}

Analyticity of $f\left[ {x;k\left( x \right)}\right]$ across requires that

\begin{equation}
\begin{array}{l}
f_1 \left( X \right) = f_2 \left( X \right) \\
f_1 ^\prime  \left( X \right) = f_2 ^\prime  \left( X \right) \\
\end{array}
\label{eq8}
\end{equation}

\noindent
We hence arrive in the system of equations

\begin{equation}
\begin{array}{l}
\left[ {\begin{array}{*{20}c}
{A\left( {X;k_2 } \right)} & {B\left( {X;k_2 } \right)}  \\
{A_x \left( {X;k_2 } \right)} & {B_x \left( {X;k_2 } \right)}  \\
\end{array}} \right]\left\{ {\begin{array}{*{20}c}
{a_2 }  \\
{b_2 }  \\
\end{array}} \right\} = \\
\quad \quad \;\;\;\,
\left[ {\begin{array}{*{20}c}
{A\left( {X;k_1 } \right)} & {B\left( {X;k_1 } \right)}  \\
{A_x \left( {X;k_1 } \right)} & {B_x \left( {X;k_1 } \right)}  \\
\end{array}} \right]\left\{ {\begin{array}{*{20}c}
{a_1 }  \\
{b_1 }  \\
\end{array}} \right\}
\end{array}
\label{eq9}
\end{equation}

\noindent
Since the Wronskian $W$ is supposed to be non-zero by (\ref{eq4}), then (\ref{eq9})
can be solved to obtain

\begin{equation}
\begin{array}{l}
\left\{ {\begin{array}{*{20}c}
{a_2 }  \\
{b_2 }  \\
\end{array}} \right\} = \left[ {\begin{array}{*{20}c}
{A_2 } & {B_2 }  \\
{A_2 ^\prime  } & {B_2 ^\prime  }  \\
\end{array}} \right]^{ - 1} \left[ {\begin{array}{*{20}c}
{A_1 } & {B_1 }  \\
{A_1 ^\prime  } & {B_1 ^\prime  }  \\
\end{array}} \right]\left\{ {\begin{array}{*{20}c}
{a_1 }  \\
{b_1 }  \\
\end{array}} \right\} = \\
\quad \quad \;\;\;\,
\frac{1}{{W_2 }}\left[ {\begin{array}{*{20}c}
{B_2 ^\prime  } & { - B_2 }  \\
{ - A_2 ^\prime  } & {A_2 }  \\
\end{array}} \right]\left[ {\begin{array}{*{20}c}
{A_1 } & {B_1 }  \\
{A_1 ^\prime  } & {B_1 ^\prime  }  \\
\end{array}} \right]\left\{ {\begin{array}{*{20}c}
{a_1 }  \\
{b_1 }  \\
\end{array}} \right\}
\end{array}
\label{eq10}
\end{equation}

\noindent
Here, $A_i $ and $A_i ^\prime$ respectively denote $A\left( {X;k_i }
\right)$ and $A_x \left( {X;k_i }\right)$, $i = 1,2$. Similarly,
$B_i $ and $B_i ^\prime$ respectively denote $B\left( {X;k_i } \right)$ and
$B_x \left( {X;k_i } \right)$, $i =1,2$. Also, $W_2  = W\left( {X;k_2 }
\right)$. Hence, we get

\begin{equation}
\left\{ {C_2 } \right\} = \left[ {Q^{1 \to 2} } \right]\left\{ {C_1 }
\right\} = \left[ {\begin{array}{*{20}c}
{q_{11}^{1 \to 2} } & {q_{12}^{1 \to 2} }  \\
{q_{21}^{1 \to 2} } & {q_{22}^{1 \to 2} }  \\
\end{array}} \right]\left\{ {C_1 } \right\}
\label{eq11}
\end{equation}

\noindent
Here, $\left[ {Q^{1 \to 2} } \right]$ is referred
to as the jump transfer matrix with the elements

\begin{equation}
\begin{array}{*{20}c}
{q_{11}^{1 \to 2}  = \frac{{B_2 ^\prime  A_1  - B_2 A_1 ^\prime
}}{{W_2 }}} & {q_{12}^{1 \to 2}  = \frac{{B_2 ^\prime  B_1  - B_2 B_1
^\prime  }}{{W_2 }}}  \\
{q_{21}^{1 \to 2}  = \frac{{A_1 ^\prime  A_2  - A_1 A_2 ^\prime
}}{{W_2 }}} & {q_{22}^{1 \to 2}  = \frac{{B_1 ^\prime  A_2  - B_1 A_2
^\prime  }}{{W_2 }}}  \\
\end{array}
\label{eq12}
\end{equation}

\noindent
and

\begin{equation}
\left\{ {C_i } \right\} = \left\{ {\begin{array}{*{20}c}
{a_i }  \\
{b_i }  \\
\end{array}} \right\},i = 1,2
\label{eq13}
\end{equation}

Properties of jump transfer matrices have been extensively discussed
in earlier works. But for the sake of convenience we mention a few

\begin{equation}
\begin{array}{l}
\left| {Q^{1 \to 2} } \right| = \frac{{W_1 }}{{W_2 }} \\
\left[ {Q^{m \to n} } \right] = \left[ {Q^{n - 1 \to n} }
\right]\left[ {Q^{n - 2 \to n - 1} } \right] \cdots \left[ {Q^{m + 1
\to m + 2} } \right]\left[ {Q^{m \to m + 1} } \right] \\
\left[ {Q^{m \to n} } \right]^{ - 1}  = \left[ {Q^{n \to m} } \right]
\\
\end{array}
\label{eq14}
\end{equation}

\noindent
corresponding respectively to the determinant, combination, and
inversion properties. Moreover, we readily notice that
$\left| {Q^{m \to n} } \right| = {{W_m } \mathord{\left/
{\vphantom {{W_m } {W_n }}} \right.
\kern-\nulldelimiterspace} {W_n }}$. It is clear that a
given transfer matrix $\left[ {Q^{1 \to 2} }
\right]$ is not invertible unless the Wronskian (\ref{eq4})
does not vanish. The combination property explains how to obtain the
total transfer matrix over a number of finite jumps, among which the
eigenvalue function $k\left( x \right)$ is
constant.

\subsection{Differential Transfer Matrix}

Now, we let $k\left( x \right)$ be a smooth
function of $x$. Within the infinitesimal
neighborhood of any given point such as $x = X$,
the eigenvalue function $k\left( x \right)$ will
undergo a first-order change from $k_1  = k\left( X
\right)$ to $k_2  = k\left( {X + \Delta x}
\right)$. We may then define $k_2  = k_1  + \Delta
k$, where $\Delta k$ represents a small
change in the eigenvalue. If $\Delta x$ is small,
then we may neglect the variations of $k\left( x
\right)$ within $\left[ {X,X + \Delta x}
\right]$. The corresponding first order change in the vector
$\left\{ C \right\}$ across $x =
X$ will be clearly given by

\begin{equation}
\left\{ {\Delta C} \right\} \approx \frac{1}{{\Delta x}}\left( {\left[
{Q^{1 \to 2} } \right] - \left[ I \right]} \right)\left\{ {C_1 }
\right\}
\label{eq15}
\end{equation}

\noindent
where the approximation becomes exact if we let $\Delta
x$ approach zero. Thereby, we get

\begin{equation}
d\left\{ {C\left( x \right)} \right\} = \left[ {U\left( x \right)}
\right]\left\{ {C\left( x \right)} \right\}dx
\label{eq16}
\end{equation}

\noindent
in which $\left\{ {C\left( x \right)} \right\}$ is
the envelope vector function, and the Kernel matrix $\left[
{U\left( x \right)} \right]$ is

\begin{equation}
\left[ {U\left( x \right)} \right] = \left[ {\begin{array}{*{20}c}
{u_{11} \left( x \right)} & {u_{12} \left( x \right)}  \\
{u_{21} \left( x \right)} & {u_{22} \left( x \right)}  \\
\end{array}} \right] = \lim _{\Delta x \to 0} \frac{1}{{\Delta
x}}\left( {\left[ {Q^{1 \to 2} } \right] - \left[ I \right]} \right)
\label{eq17}
\end{equation}

We notice that in order to obtain the correct solution to the Kernel
matrix $\left[ {U\left( x \right)} \right]$, one
needs to make the replacements $G_1  = G$,
$G_1 ^\prime   = G_x $, and $G_2  = G +
G_k \Delta k$, $G_2 ^\prime   = G_x  + G_{xk}
\Delta k$, where represents either of $A$
or $B$. Here, subscripts refer to partial
derivatives in the sense that $G_x  = {{\partial G}
\mathord{\left/
{\vphantom {{\partial G} {\partial x}}} \right.
\kern-\nulldelimiterspace} {\partial x}}$, $G_k
= {{\partial G} \mathord{\left/
{\vphantom {{\partial G} {\partial k}}} \right.
\kern-\nulldelimiterspace} {\partial k}}$,
$G_{xk}  = {{\partial ^2 G} \mathord{\left/
{\vphantom {{\partial ^2 G} {\partial x\partial k}}} \right.
\kern-\nulldelimiterspace} {\partial x\partial k}}$. After
doing some algebra and simplification we get the complete form for the
elements of the Kernel matrix $\left[ {U\left( x \right)}
\right]$ as

\begin{equation}
\begin{array}{*{20}c}
{u_{11}  = k'\frac{{A_{xk} B - A_k B_x }}{W}} & {u_{12}  =
k'\frac{{B_{xk} B - B_k B_x }}{W}}  \\
{u_{21}  = k'\frac{{A_k A_x  - A_{xk} A}}{W}} & {u_{22}  =
k'\frac{{A_x B_k  - AB_{xk} }}{W}}  \\
\end{array}
\label{eq18}
\end{equation}

\noindent
Here, $k' = {{\partial k\left( x \right)} \mathord{\left/
{\vphantom {{\partial k\left( x \right)} {\partial x}}} \right.
\kern-\nulldelimiterspace} {\partial x}}$ and $W
= AB_x  - A_x B$. In the fully expanded form we have

\begin{equation}
\begin{array}{l}
u_{11} \left( x \right) = \frac{{k'\left( x \right)}}{{W\left( x
\right)}}\left\{ {\frac{{\partial ^2 A\left[ {x;k\left( x \right)}
\right]}}{{\partial x\partial k}}B\left[ {x;k\left( x \right)} \right]
- \frac{{\partial A\left[ {x;k\left( x \right)} \right]}}{{\partial
k}}\frac{{\partial B\left[ {x;k\left( x \right)} \right]}}{{\partial
x}}} \right\} \\
u_{12} \left( x \right) = \frac{{k'\left( x \right)}}{{W\left( x
\right)}}\left\{ {\frac{{\partial ^2 B\left[ {x;k\left( x \right)}
\right]}}{{\partial x\partial k}}B\left[ {x;k\left( x \right)} \right]
- \frac{{\partial B\left[ {x;k\left( x \right)} \right]}}{{\partial
k}}\frac{{\partial B\left[ {x;k\left( x \right)} \right]}}{{\partial
x}}} \right\} \\
u_{21} \left( x \right) = \frac{{k'\left( x \right)}}{{W\left( x
\right)}}\left\{ {\frac{{\partial A\left[ {x;k\left( x \right)}
\right]}}{{\partial k}}\frac{{\partial A\left[ {x;k\left( x \right)}
\right]}}{{\partial x}} - \frac{{\partial ^2 A\left[ {x;k\left( x
\right)} \right]}}{{\partial x\partial k}}A\left[ {x;k\left( x \right)}
\right]} \right\} \\
u_{22} \left( x \right) = \frac{{k'\left( x \right)}}{{W\left( x
\right)}}\left\{ {\frac{{\partial A\left[ {x;k\left( x \right)}
\right]}}{{\partial x}}\frac{{\partial B\left[ {x;k\left( x \right)}
\right]}}{{\partial k}} - A\left[ {x;k\left( x \right)}
\right]\frac{{\partial ^2 B\left[ {x;k\left( x \right)}
\right]}}{{\partial x\partial k}}} \right\} \\
\end{array}
\label{eq19}
\end{equation}

\noindent
where

\begin{equation}
W\left( x \right) = A\left[ {x;k\left( x \right)}
\right]\frac{{\partial B\left[ {x;k\left( x \right)}
\right]}}{{\partial x}} - \frac{{\partial A\left[ {x;k\left( x \right)}
\right]}}{{\partial x}}B\left[ {x;k\left( x \right)} \right]
\label{eq20}
\end{equation}

A general solution to (\ref{eq2})
is given by (\ref{eq5}) with

\begin{equation}
\left\{ {C'\left( x \right)} \right\} = \left[ {U\left( x \right)}
\right]\left\{ {C\left( x \right)} \right\}
\label{eq21}
\end{equation}

\noindent
in which the elements of the Kernel matrix $\left[ {U\left(
x \right)} \right]$ are given in (\ref{eq18}). Interestingly,
(\ref{eq21}) allows an analytical solution through perturbation theory as \cite{ref2,ref3}

\begin{equation}
\begin{array}{l}
\left\{ {C\left( {x_2 } \right)} \right\} = \left\{ {C\left( {x_1 }
\right)} \right\} + \int\limits_{x_1 }^{x_2 } {\left[ {U\left( {y_0 }
\right)} \right]\left\{ {C\left( {y_0 } \right)} \right\}dy_0 }  \\
\quad \quad \;\; + \int\limits_{x_1 }^{x_2 } {\int\limits_{x_1 }^{y_1
} {\left[ {U\left( {y_1 } \right)} \right]\left[ {U\left( {y_0 }
\right)} \right]\left\{ {C\left( {y_0 } \right)} \right\}dy_0 dy_1 } }
\\
\quad \quad \;\; + \int\limits_{x_1 }^{x_2 } {\int\limits_{x_1 }^{y_2
} {\int\limits_{x_1 }^{y_1 } {\left[ {U\left( {y_2 } \right)}
\right]\left[ {U\left( {y_1 } \right)} \right]\left[ {U\left( {y_0 }
\right)} \right]\left\{ {C\left( {y_0 } \right)} \right\}dy_0 dy_1 dy_2
} } }  +  \cdots  \\
\end{array}
\label{eq22}
\end{equation}

Often (\ref{eq22}) is written symbolically as

\begin{equation}
\begin{array}{l}
\left\{ {C\left( {x_2 } \right)} \right\} = \mathbb{T}\exp \left\{
{\int\limits_{x_1 }^{x_2 } {\left[ {U\left( x \right)} \right]dx} }
\right\}\left\{ {C\left( {x_1 } \right)} \right\} \\
\quad \quad \;\; = \left\{ {C\left( {x_1 } \right)} \right\} +
\int\limits_{x_1 }^{x_2 } {\left[ {U\left( {y_0 } \right)}
\right]\left\{ {C\left( {y_0 } \right)} \right\}dy_0 }  \\
\quad \quad \quad \;\; + \frac{1}{{2!}}\int\limits_{x_1 }^{x_2 }
{\int\limits_{x_1 }^{x_2 } {\mathbb{T}\left[ {U\left( {y_1 } \right)}
\right]\left[ {U\left( {y_0 } \right)} \right]\left\{ {C\left( {y_0 }
\right)} \right\}dy_0 dy_1 } }  \\
\quad \quad \quad \;\; + \frac{1}{{3!}}\int\limits_{x_1 }^{x_2 }
{\int\limits_{x_1 }^{x_2 } {\int\limits_{x_1 }^{x_2 } {\mathbb{T}\left[
{U\left( {y_2 } \right)} \right]\left[ {U\left( {y_1 } \right)}
\right]\left[ {U\left( {y_0 } \right)} \right]\left\{ {C\left( {y_0 }
\right)} \right\}dy_0 dy_1 dy_2 } } }  +  \cdots  \\
\quad \quad \;\; = \left[ {Q^{x_1  \to x_2 } } \right]\left\{
{C\left( {x_1 } \right)} \right\} \\
\end{array}
\label{eq23}
\end{equation}

\noindent
where $\left[ {Q^{x_1  \to x_2 } } \right]$ is the
transfer matrix from $x_1 $ to $x_2
$ and $\exp \left(  \cdot  \right)$ being
the matrix exponentiation

\begin{equation}
\exp \left[ M \right] = \left[ I \right] + \sum\limits_{n = 1}^\infty
{\frac{1}{{n!}}} \left[ M \right]^n
\label{eq24}
\end{equation}

Furthermore, $\mathbb{T}$ is the Dyson's ordering
operator defined as

\begin{equation}
\mathbb{T}\left[ {U\left( a \right)} \right]\left[ {U\left( b \right)}
\right] = \left\{ {\begin{array}{*{20}c}
{\left[ {U\left( a \right)} \right]\left[ {U\left( b \right)}
\right],} & {a > b}  \\
{\left[ {U\left( b \right)} \right]\left[ {U\left( a \right)}
\right],} & {a < b}  \\
\end{array}} \right.
\label{eq25}
\end{equation}

There are few known sufficient conditions \cite{ref3,ref13}, which rarely happen and
under which $\mathbb{T}$ may be dropped exactly to reach

\begin{equation}
\left\{ {C\left( {x_2 } \right)} \right\} = \exp \left\{
{\int\limits_{x_1 }^{x_2 } {\left[ {U\left( x \right)} \right]dx} }
\right\}\left\{ {C\left( {x_1 } \right)} \right\}
\label{eq26}
\end{equation}

The most important sufficient condition includes the case when the
Kernel matrix commutes with itself as $\left[ {U\left( a
\right)} \right]\left[ {U\left( b \right)} \right] = \left[ {U\left( b
\right)} \right]\left[ {U\left( a \right)} \right]$ for all
given $a$ and $b$. This is better known as the Lappo-Danilevskii \cite{ref13a} 
criterion, and is also generalized by Fedorov \cite{ref13b}. Evidently,
this condition applies to all constant as well as diagonal Kernel
matrices. It is therefore a challenge to find construct the proper
extended eigenvalue equation in such a way that the corresponding
Kernel matrix meets any of these sufficient criteria. If possible, then
exact and explicit closed form solutions to (\ref{eq2})
are found analytically by using (\ref{eq26}) instead of (\ref{eq23}).

While in general, (\ref{eq26}) is merely an approximation to the exact
solution (\ref{eq23}), nevertheless, in any case the determinant and trace of
$\left[ {Q^{x_1  \to x_2 } } \right]$ would be
preserved exactly, meaning that the eigenvalues of the transfer matrix
$\left[ {Q^{x_1  \to x_2 } } \right]$ would not be
affected at least. This allows us to formulate a very convenient
approximate numerical solution to the extended problem (\ref{eq2}).

\subsection{Properties of Differential Transfer Matrix}

Properties of the transfer matrix $\left[ {Q^{x_1  \to x_2
} } \right]$ as defined in (\ref{eq23}) are very much similar to
those of the jump transfer matrix
(\ref{eq14}), given by

\begin{equation}
\begin{array}{l}
\left[ {Q^{x_1  \to x_1 } } \right] = \left[ I \right] \\
\left| {Q^{x_1  \to x_2 } } \right| = \frac{{W_1 }}{{W_2 }} \\
\left[ {Q^{x_1  \to x_2 } } \right] = \left[ {Q^{x_3  \to x_2 } }
\right]\left[ {Q^{x_1  \to x_3 } } \right] \\
\left[ {Q^{x_1  \to x_2 } } \right]^{ - 1}  = \left[ {Q^{x_2  \to x_1
} } \right] \\
\end{array}
\label{eq27}
\end{equation}

Properties of the transfer matrix in which $W_i  = W\left[
{x_i ;k_i } \right] = W\left[ {x_i ;k\left( {x_i } \right)} \right],i =
1,2$. The first property, the identity property readily
follows by definition in (\ref{eq23}). The determinant property can be
observed by noting that Dyson's operator has no effect on the determinant,
and thus

\begin{equation}
\begin{array}{l}
\left| {Q^{x_1  \to x_2 } } \right| = \left| {\mathbb{T}\exp \left\{
{\int\limits_{x_1 }^{x_2 } {\left[ {U\left( x \right)} \right]dx} }
\right\}} \right| = \left| {\exp \left\{ {\int\limits_{x_1 }^{x_2 }
{\left[ {U\left( x \right)} \right]dx} } \right\}} \right| \\
\quad \quad \,\,
= \exp\left( {{\mathop{\rm tr}\nolimits} \left\{ {\int\limits_{x_1 }^{x_2 }
{\left[ {U\left( x \right)} \right]dx} } \right\}} \right)
 = \exp \left( {\int\limits_{x_1 }^{x_2 } {\left[
{u_{11} \left( x \right) + u_{22} \left( x \right)} \right]dx} }
\right) \\
\end{array}
\label{eq28}
\end{equation}

Now from (\ref{eq18}) we get

\begin{equation}
\begin{array}{l}
\int\limits_{x_1 }^{x_2 } {\left[ {u_{11} \left( x \right) + u_{22}
\left( x \right)} \right]dx}  = \int\limits_{x_1 }^{x_2 }
{\frac{{A_{xk} B - AB_{xk} }}{{AB_x  - BA_x }}k'dx}  =  -
\int\limits_{x_1 }^{x_2 } {\frac{\partial }{{\partial k}}\ln
W\frac{{\partial k}}{{\partial x}}dx}  \\
\quad  =  - \int\limits_{k_1 }^{k_2 } {\frac{\partial }{{\partial
k}}\ln Wdk}  = \ln k_1  - \ln k_2  \\
\end{array}
\label{eq29}
\end{equation}

Inserting this result in (\ref{eq28})
gives the determinant property in (\ref{eq28}). Combination
and inversion properties also follow the definition
(\ref{eq23}).

\subsection{Fundamental Theorem}

In this section, we rigorously show that the differential transfer
matrix method leads to an exact solution of the differential equation
(\ref{eq2}). To show this, we start by proving the so-called \textit{Derivative
Lemma}.

\begin{lem}
The total derivative of function (\ref{eq5}) is given by

\begin{equation}
\frac{{d^n }}{{dx^n }}f\left[ {x;k\left( x \right)} \right] = a\left(
x \right)\frac{{\partial ^n }}{{\partial x^n }}A\left[ {x;k\left( x
\right)} \right] + b\left( x \right)\frac{{\partial ^n }}{{\partial x^n
}}B\left[ {x;k\left( x \right)} \right],\,\,0 \le n \le 2
\label{eq30}
\end{equation}

\end{lem}

\begin{proof}
Proof follows by direct substitution. The
case of $n = 0$ is trivial, and we first prove the
validity of (\ref{eq30}) for $n = 1$. Direct differentiation of
(\ref{eq5})
by chain rule gives

\begin{equation}
\frac{d}{{dx}}f = a'A + b'B + a\left( {A_x  + k'A_k } \right) +
b\left( {B_x  + k'B_k } \right)
\label{eq31}
\end{equation}

\noindent
where $f = aA + bB$. But we already
have obtained the derivatives $a'$ and
$b'$ from (\ref{eq21}) and (\ref{eq18}) as

\begin{equation}
\begin{array}{l}
a' = u_{11} a + u_{12} b = \frac{{k'}}{W}\left[ {\left( {A_{xk} B -
A_k B_x } \right)a + \left( {B_{xk} B - B_k B_x } \right)b} \right] \\
b' = u_{21} a + u_{22} b = \frac{{k'}}{W}\left[ {\left( {A_k A_x  -
A_{xk} A} \right)a + \left( {B_k A_x  - B_{xk} A} \right)b} \right] \\
\end{array}
\label{eq32}
\end{equation}

\noindent
in which the elements $u_{ij} ,i,j =
1,2$ of the Kernel matrix are given in (\ref{eq18}). After some minor algebra we get

\begin{equation}
\frac{d}{{dx}}f = aA_x  + bB_x
\label{eq33}
\end{equation}

\noindent
To show the correctness of (\ref{eq30}) for $n = 2$, we take the derivative again with
respect to $x$ from both sides of (\ref{eq30}), thus giving

\begin{equation}
\begin{array}{l}
\frac{{d^2 }}{{dx^2 }}f = \frac{d}{{dx}}\left( {aA_x  + bB_x }
\right) \\
\quad  = a'A_x  + b'B_x  + a\frac{d}{{dx}}A_x  + b\frac{d}{{dx}}B_x
\\
\quad  = \left( {u_{11} a + u_{12} b} \right)A_x  + \left( {u_{21} a
+ u_{22} b} \right)B_x  \\
\quad  + a\left( {A_{xx}  + k'A_{xk} } \right) + b\left( {B_{xx}  +
k'B_{xk} } \right) \\
\end{array}
\label{eq34}
\end{equation}

\noindent
Further substitution of (\ref{eq32})
in (\ref{eq34}) gives

\begin{equation}
\begin{array}{l}
\frac{{d^2 }}{{dx^2 }}f\; = \frac{{k'}}{{AB_x  - BA_x }}\left[
{\left( {A_{xk} B - A_k B_x } \right)a + \left( {B_{xk} B - B_k B_x }
\right)b} \right]A_x  \\
\quad \quad \;\; + \frac{{k'}}{{AB_x  - BA_x }}\left[ {\left( {A_k
A_x  - A_{xk} A} \right)a + \left( {B_k A_x  - B_{xk} A} \right)b}
\right]B_x  \\
\quad \quad \;\; + a\left( {A_{xx}  + k'A_{xk} } \right) + b\left(
{B_{xx}  + k'B_{xk} } \right) \\
\end{array}
\label{eq35}
\end{equation}

\noindent
Here, we have used the definition of the Wronskian
$W = AB_x  - BA_x $ from (\ref{eq4}). Eventually, after some
algebra we arrive at the final result

\begin{equation}
\frac{{d^2 }}{{dx^2 }}f\; = aA_{xx}  + bB_{xx}
\label{eq36}
\end{equation}

\noindent
This completes the proof

\end{proof}

We are now in a position to present the \textit{Fundamental Theorem of
Differential Transfer Matrix Method }as follows.

\begin{thm}
The solution to (\ref{eq2}) having the form (\ref{eq5})
with (\ref{eq21}) is exact.
\end{thm}

\begin{proof}
To show the validity of the statement, we
start by plugging in (\ref{eq5}) directly into (\ref{eq2}).

\begin{equation}
\begin{array}{l}
\mathbb{L}_k f\left[ {x;k} \right] = \mathbb{L}_k \left\{ {aA\left[ {x;k}
\right] + bB\left[ {x;k} \right]} \right\} \\
\quad \quad \quad \quad \; = a\mathbb{H}_k A\left[ {x;k} \right] +
b\mathbb{H}_k B\left[ {x;k} \right] \\
\end{array}
\label{eq37}
\end{equation}

\noindent where $\mathbb{H}_k $ and
$\mathbb{L}_k $ are respectively defined in
(\ref{eq1}) and (\ref{eq2}). Here, the explicit
dependence of the operator $\mathbb{L}_k$ and functions $a$,
$b$, and $k$ on $x$ is not shown. But by assumption we have
$\mathbb{H}_k A\left[ {x;k} \right] = \mathbb{H}_k B\left[ {x;k}
\right] = 0$ and hence the assertion.
\end{proof}

All remains is to force the initial or boundary conditions, and we
here mention the treatment of both types of conditions.

\subsection{Initial Conditions}

Without loss of generality, we may assume that initial conditions are
known at some point like $x = c$. Suppose that
$f\left( c \right)$ and $f'\left( c
\right)$ are known. Then by derivative lemma (\ref{eq30})
we have

\begin{equation}
f'\left( x \right) = A_x \left[ {x;k\left( x \right)} \right]a\left( x
\right) + B_x \left[ {x;k\left( x \right)} \right]b\left( x \right)
\label{eq38}
\end{equation}

\noindent
Therefore we get

\begin{equation}
\left\{ {\begin{array}{*{20}c}
{f\left( c \right)}  \\
{f'\left( c \right)}  \\
\end{array}} \right\} = \left[ {\begin{array}{*{20}c}
{A\left[ {c;k\left( c \right)} \right]} & {B\left[ {c;k\left( c
\right)} \right]}  \\
{A_x \left[ {c;k\left( c \right)} \right]} & {B_x \left[ {c;k\left(
c \right)} \right]}  \\
\end{array}} \right]\left\{ {\begin{array}{*{20}c}
{a\left( c \right)}  \\
{b\left( c \right)}  \\
\end{array}} \right\}
\label{eq39}
\end{equation}

\noindent
Thus we may obtain the initial vector $\left\{ {C\left( c
\right)} \right\}$ as

\begin{equation}
\left\{ {C\left( c \right)} \right\} = \left[ {\begin{array}{*{20}c}
{A\left[ {c;k\left( c \right)} \right]} & {B\left[ {c;k\left( c
\right)} \right]}  \\
{A_x \left[ {c;k\left( c \right)} \right]} & {B_x \left[ {c;k\left(
c \right)} \right]}  \\
\end{array}} \right]^{ - 1} \left\{ {\begin{array}{*{20}c}
{f\left( c \right)}  \\
{f'\left( c \right)}  \\
\end{array}} \right\}
\label{eq40}
\end{equation}

\noindent
with the solution

\begin{equation}
\left\{ {C\left( c \right)} \right\} = \frac{1}{{W\left[ {c;k\left( c
\right)} \right]}}\left[ {\begin{array}{*{20}c}
{B_x \left[ {c;k\left( c \right)} \right]} & { - B\left[ {c;k\left(
c \right)} \right]}  \\
{ - A_x \left[ {c;k\left( c \right)} \right]} & {A\left[ {c;k\left(
c \right)} \right]}  \\
\end{array}} \right]\left\{ {\begin{array}{*{20}c}
{f\left( c \right)}  \\
{f'\left( c \right)}  \\
\end{array}} \right\}
\label{eq41}
\end{equation}

Now by (\ref{eq23}) we have the full solution to the envelope vector as

\begin{equation}
\begin{array}{l}
\left\{ {C\left( x \right)} \right\} = \mathbb{T}\exp \left\{
{\int\limits_c^x {\left[ {U\left( y \right)} \right]dy} }
\right\}\left\{ {C\left( c \right)} \right\}
\\
\quad \quad \,\,
\approx \exp \left\{
{\int\limits_c^x {\left[ {U\left( y \right)} \right]dy} }
\right\}\left\{ {C\left( c \right)} \right\}
\end{array}
\label{eq42}
\end{equation}

\noindent
This allows us to obtain the final solution to the initial value
problem (\ref{eq2}) as

\begin{equation}
f\left( x \right) = \left[ {\begin{array}{*{20}c}
{A\left[ {x;k\left( x \right)} \right]} & {B\left[ {x;k\left( x
\right)} \right]}  \\
\end{array}} \right]\left\{ {C\left( x \right)} \right\}
\label{eq43}
\end{equation}

\subsection{Boundary Conditions}

Without loss of generality, we may assume that boundary conditions are
known at two points like $x = c_1 ,c_2 $. Suppose
that $f\left( {c_1 } \right)$ and
$f\left( {c_2 } \right)$ are known. Then by
(\ref{eq5}) we have

\begin{equation}
\begin{array}{l}
f\left( {c_1 } \right) = \left[ {\begin{array}{*{20}c}
{A\left[ {c_1 ;k\left( {c_1 } \right)} \right]} & {B\left[ {c_1
;k\left( {c_1 } \right)} \right]}  \\
\end{array}} \right]\left\{ {C\left( {c_1 } \right)} \right\} \\
f\left( {c_2 } \right) = \left[ {\begin{array}{*{20}c}
{A\left[ {c_2 ;k\left( {c_2 } \right)} \right]} & {B\left[ {c_2
;k\left( {c_2 } \right)} \right]}  \\
\end{array}} \right]\left\{ {C\left( {c_2 } \right)} \right\} \\
\end{array}
\label{eq44}
\end{equation}

But from (\ref{eq23}) we have

\begin{equation}
\left\{ {C\left( {c_2 } \right)} \right\} = \left[ {Q^{c_1  \to c_2 }
} \right]\left\{ {C\left( {c_1 } \right)} \right\}
\label{eq45}
\end{equation}

Thus by combining (\ref{eq44})
and (\ref{eq45}) we may obtain the system of equations

\begin{equation}
\begin{array}{l}
\left\{ {\begin{array}{*{20}c}
{f\left( {c_1 } \right)}  \\
{f\left( {c_2 } \right)}  \\
\end{array}} \right\} =
\\
\quad
\left[ {\begin{array}{*{20}c}
{A\left[ {c_1 ;k\left( {c_1 } \right)} \right]} & {B\left[ {c_1
;k\left( {c_1 } \right)} \right]}  \\
{q_{11} A\left[ {c_2 ;k\left( {c_2 } \right)} \right] + q_{21}
B\left[ {c_2 ;k\left( {c_2 } \right)} \right]} & {q_{12} A\left[ {c_2
;k\left( {c_2 } \right)} \right] + q_{22} B\left[ {c_2 ;k\left( {c_2 }
\right)} \right]}  \\
\end{array}} \right] \times \\ \quad
\left\{ {C\left( {c_1 } \right)} \right\}
\end{array}
\label{eq46}
\end{equation}

\noindent
with the solution

\begin{equation}
\begin{array}{l}
\left\{ {C\left( {c_1 } \right)} \right\} = \\
\quad
\left[
{\begin{array}{*{20}c}
{A\left[ {c_1 ;k\left( {c_1 } \right)} \right]} & {B\left[ {c_1
;k\left( {c_1 } \right)} \right]}  \\
{q_{11} A\left[ {c_2 ;k\left( {c_2 } \right)} \right] + q_{21}
B\left[ {c_2 ;k\left( {c_2 } \right)} \right]} & {q_{12} A\left[ {c_2
;k\left( {c_2 } \right)} \right] + q_{22} B\left[ {c_2 ;k\left( {c_2 }
\right)} \right]}  \\
\end{array}} \right]^{ - 1} \times \\
\quad \left\{ {\begin{array}{*{20}c}
{f\left( {c_1 } \right)}  \\
{f\left( {c_2 } \right)}  \\
\end{array}} \right\} \\
\left\{ {C\left( {c_2 } \right)} \right\} = \left[ {Q^{c_1  \to c_2 }
} \right]\left\{ {C\left( {c_1 } \right)} \right\} \\
\end{array}
\label{eq47}
\end{equation}

The rest is the same and similar to (\ref{eq42})

\begin{equation}
\begin{array}{l}
\left\{ {C\left( x \right)} \right\} = \mathbb{T}\exp \left\{
{\int\limits_{c_1 }^x {\left[ {U\left( y \right)} \right]dy} }
\right\}\left\{ {C\left( {c_1 } \right)} \right\} = \mathbb{T}\exp
\left\{ { - \int\limits_x^{c_2 } {\left[ {U\left( y \right)} \right]dy}
} \right\}\left\{ {C\left( {c_2 } \right)} \right\} \\
\quad \quad \; \approx \exp \left\{ {\int\limits_{c_1 }^x {\left[
{U\left( y \right)} \right]dy} } \right\}\left\{ {C\left( {c_1 }
\right)} \right\} \approx \exp \left\{ { - \int\limits_x^{c_2 } {\left[
{U\left( y \right)} \right]dy} } \right\}\left\{ {C\left( {c_2 }
\right)} \right\} \\
\end{array}
\label{eq48}
\end{equation}

\section{Examples}

In this section, we present some application examples describing the
details of the method.

\subsection{Wave Equation}

Numerous physical problems are described via the simple second-order
equation

\begin{equation}
\psi ''\left( x \right) + k^2 \left( x \right)\psi \left( x \right) =
0
\label{eq49}
\end{equation}

\noindent
The equation (\ref{eq49}) is known to have no explicit solution for arbitrary eigenvalue
function $k\left( x \right)$, which in the
literature is usually referred to as the wavenumber function. The only
existing analytical solution to (\ref{eq49})
is the very well-known WKB approximation. Here, the corresponding
operators read

\begin{equation}
\mathbb{H}_k  = \frac{{\partial ^2 }}{{\partial x^2 }} + k^2 \quad \quad
\mathbb{L}_{k\left( x \right)}  = \frac{{d^2 }}{{dx^2 }} + k^2 \left( x
\right)
\label{eq50}
\end{equation}

\noindent
In case of constant wavenumber, any solution of (\ref{eq49})
is given by linear combinations of exponential functions $\exp
\left( { \pm ikx} \right)$, and therefore the eigenfunctions
are readily found to be

\begin{equation}
A\left[ {x;k\left( x \right)} \right] = \exp \left[ { - ixk\left( x
\right)} \right]\quad B\left[ {x;k\left( x \right)} \right] = \exp
\left[ { + ixk\left( x \right)} \right]
\label{eq51}
\end{equation}

\noindent
with the Wronskian $W\left[ {x;k\left( x \right)} \right] =
2ik\left( x \right)$. To find a solution to (\ref{eq49})
it is therefore sufficient to find the Kernel matrix
$\left[ {U\left( x \right)} \right]$, which is

\begin{equation}
\left[ {U\left( x \right)} \right] = \frac{{k'\left( x
\right)}}{{2k\left( x \right)}}\left[ {\begin{array}{*{20}c}
{ - 1 + 2ixk\left( x \right)} & {\exp \left[ { + 2ixk\left( x
\right)} \right]}  \\
{\exp \left[ { - 2ixk\left( x \right)} \right]} & { - 1 -
2ixk\left( x \right)}  \\
\end{array}} \right]
\label{eq52}
\end{equation}

Numerical and analytical solutions of (\ref{eq49})
using the Kernel matrix (\ref{eq52}) has been extensively
studied in the literature through past years.

\subsection{Airy's Equation}

One of the known problems with the above differential transfer matrix
solution of (\ref{eq49}) is occurrence of singularities at which $k^2 \left( x
\right)$ changes sign and the Wronskian vanishes, better known
as returning points. The singularity is clear in (\ref{eq52})
as the denominator. Even the approximate WKB solution fails near
returning points. It is a common practice however to expand the
solution near the returning points via Airy functions. Without loss of
generality we may assume that the singularity is located at
$x = 0$. In this case, we present a slightly
modified form of (\ref{eq49}) as

\begin{equation}
\psi ''\left( x \right) - k^3 \left( x \right)x\psi \left( x \right) =
0
\label{eq53}
\end{equation}

\noindent
with the Airy wavenumber satisfying $k\left( 0 \right) \ne0$.
Hence,

\begin{equation}
\mathbb{H}_k  = \frac{{\partial ^2 }}{{\partial x^2 }} - k^3 x\quad
\quad \mathbb{L}_{k\left( x \right)}  = \frac{{d^2 }}{{dx^2 }} - k^3
\left( x \right)x
\label{eq54}
\end{equation}

In case of constant wavenumber, any solution of (\ref{eq53})
is given by linear combinations of Airy functions
${\mathop{\rm Ai}\nolimits} \left( {kx} \right)$
and ${\mathop{\rm Bi}\nolimits} \left( {kx}
\right)$, and thus the eigenfunctions take the form

\begin{equation}
A\left[ {x;k\left( x \right)} \right] = {\mathop{\rm Ai}\nolimits}
\left[ {xk\left( x \right)} \right]\quad B\left[ {x;k\left( x \right)}
\right] = {\mathop{\rm Bi}\nolimits} \left[ {xk\left( x \right)}
\right]
\label{eq55}
\end{equation}

\noindent
having the constant Wronskian $W\left[ {x;k\left( x
\right)} \right] = {1 \mathord{\left/
{\vphantom {1 \pi }} \right.
\kern-\nulldelimiterspace} \pi }$. This choice of basis,
after some simplification, leads to the following Kernel matrix
elements

\begin{equation}
\begin{array}{l}
u_{11}  =  + \pi k'\left\{ {k^2 x^2 {\mathop{\rm Ai}\nolimits} \left(
{kx} \right){\mathop{\rm Bi}\nolimits} \left( {kx} \right) +
{\mathop{\rm A}\nolimits} i'\left( {kx} \right)\left[ {{\mathop{\rm
Bi}\nolimits} \left( {kx} \right) - kx{\mathop{\rm B}\nolimits}
i'\left( {kx} \right)} \right]} \right\} \\
u_{12}  = \pi k'\left\{ {{\mathop{\rm B}\nolimits} i'\left( {kx}
\right)\left[ {{\mathop{\rm Bi}\nolimits} \left( {kx} \right) -
xk{\mathop{\rm B}\nolimits} i'\left( {kx} \right)} \right] + k^2 x^2
{\mathop{\rm Bi}\nolimits} ^2 \left( {kx} \right)} \right\} \\
u_{21}  = \pi k'\left\{ {{\mathop{\rm A}\nolimits} i'\left( {kx}
\right)\left[ {xk{\mathop{\rm A}\nolimits} i'\left( {kx} \right) -
{\mathop{\rm Ai}\nolimits} \left( {kx} \right)} \right] - k^2 x^2
{\mathop{\rm Ai}\nolimits} ^2 \left( {kx} \right)} \right\} \\
u_{22}  =  - \pi k'\left\{ {k^2 x^2 {\mathop{\rm Ai}\nolimits} \left(
{kx} \right){\mathop{\rm Bi}\nolimits} \left( {kx} \right) +
{\mathop{\rm B}\nolimits} i'\left( {kx} \right)\left[ {{\mathop{\rm
Ai}\nolimits} \left( {kx} \right) - kx{\mathop{\rm A}\nolimits}
i'\left( {kx} \right)} \right]} \right\} \\
\end{array}
\label{eq56}
\end{equation}

\noindent
which is clearly non-singular at $x = 0$. Here,
dependence of the eigenvalue function $k\left( x
\right)$ on $x$ is not displayed for the
sake of convenience.

\subsection{Bessel's Equation}

In the domain of fiber optics having cylindrical symmetry and after
proper transformations, the radial component of the wave equation takes
the form

\begin{equation}
x^2 \psi ''\left( x \right) + x\psi '\left( x \right) + \left[ {k^2
\left( x \right)x^2  - \nu ^2 } \right]\psi \left( x \right) = 0
\label{eq57}
\end{equation}

\noindent
Hence

\begin{equation}
\mathbb{H}_k  = x^2 \frac{{\partial ^2 }}{{\partial x^2 }} +
x\frac{\partial }{{\partial x}} + \left( {k^2 x^2  - \nu ^2 }
\right)\quad \quad \mathbb{L}_{k\left( x \right)}  = x^2 \frac{{d^2
}}{{dx^2 }} + x\frac{d}{{dx}} + \left[ {k^2 \left( x \right)x^2  - \nu
^2 } \right]
\label{eq58}
\end{equation}

In case of constant wavenumber, any solution of (\ref{eq57})
is given by any linear combinations of Bessel and Neumann functions
$J_\nu  \left( {kx} \right)$ and $N_\nu
\left( {kx} \right)$, or Hankel functions $H_\nu
^{\left( 1 \right)} \left( {kx} \right)$ and $H_\nu
 ^{\left( 2 \right)} \left( {kx} \right)$. For the first pair
the eigenfunctions take the form

\begin{equation}
A\left[ {x;k\left( x \right)} \right] = J_\upsilon  \left[ {xk\left( x
\right)} \right]\quad B\left[ {x;k\left( x \right)} \right] =
N_\upsilon  \left[ {xk\left( x \right)} \right]
\label{eq59}
\end{equation}

\noindent
having the Wronskian $W\left[ {x;k\left( x \right)} \right]
= {2 \mathord{\left/
{\vphantom {2 {\pi x}}} \right.
\kern-\nulldelimiterspace} {\pi x}}$. The elements of the
corresponding Kernel matrix are

\begin{equation}
\begin{array}{l}
u_{11}  =  + \frac{{x\pi k'}}{2}\left\{ {\left[ {kx\left[ {J_{\nu  -
2} \left( {kx} \right) - 2J_\nu  \left( {kx} \right) + J_{\nu  + 2}
\left( {kx} \right)} \right] +  2\left[ {J_{\nu  - 1} \left( {kx}
\right) - J_{\nu  + 1} \left( {kx} \right)} \right]} \right]N_\nu
\left( {kx} \right)} \right. \\
\quad \left. {\quad \quad \quad \;\; - kx\left[ {J_{\nu  - 1} \left(
{kx} \right) - J_{\nu  + 1} \left( {kx} \right)} \right]\left[ {N_{\nu
- 1} \left( {kx} \right) - N_{\nu  + 1} \left( {kx} \right)} \right]}
\right\} \\
u_{12}  =  - \frac{{x\pi k'}}{2}\left\{ {\left[ {kx\left[ {N_{\nu  -
2} \left( {kx} \right) - 2N_\nu  \left( {kx} \right) + N_{\nu  + 2}
\left( {kx} \right)} \right] + 2\left[ {N_{\nu  - 1} \left( {kx}
\right) - N_{\nu  + 1} \left( {kx} \right)} \right]} \right]N_\nu
\left( {kx} \right)} \right. \\
\quad \left. {\quad \quad \quad \;\; + kx\left[ {N_{\nu  - 1} \left(
{kx} \right) - N_{\nu  + 1} \left( {kx} \right)} \right]^2 } \right\}
\\
u_{21}  =  + \frac{{x\pi k'}}{2}\left\{ {\left[ {kx\left[ {J_{\nu  -
2} \left( {kx} \right) - 2J_\nu  \left( {kx} \right) + J_{\nu  + 2}
\left( {kx} \right)} \right] + 2\left[ {J_{\nu  - 1} \left( {kx}
\right) - J_{\nu  + 1} \left( {kx} \right)} \right]} \right]J_\nu
\left( {kx} \right)} \right. \\
\quad \left. {\quad \quad \quad \;\; + kx\left[ {J_{\nu  - 1} \left(
{kx} \right) - J_{\nu  + 1} \left( {kx} \right)} \right]^2 } \right\}
\\
u_{22}  =  - \frac{{x\pi k'}}{2}\left\{ {\left[ {kx\left[ {N_{\nu  -
2} \left( {kx} \right) - 2N_\nu  \left( {kx} \right) + N_{\nu  + 2}
\left( {kx} \right)} \right] + 2\left[ {N_{\nu  - 1} \left( {kx}
\right) - N_{\nu  + 1} \left( {kx} \right)} \right]} \right]J_\nu
\left( {kx} \right)} \right. \\
\quad \left. {\quad \quad \quad \;\; - kx\left[ {J_{\nu  - 1} \left(
{kx} \right) - J_{\nu  + 1} \left( {kx} \right)} \right]\left[ {N_{\nu
- 1} \left( {kx} \right) - N_{\nu  + 1} \left( {kx} \right)} \right]}
\right\} \\
\end{array}
\label{eq60}
\end{equation}

As it can be seen, all elements of the Kernel matrix vanish at
$x = 0$; this ensures cylindrical symmetry of the
radial function.

It should be also noticed that since Bessel functions with non-integer
order are well-defined \cite{ref14}

\begin{equation}
J_\alpha  \left( x \right) = \sum\limits_{m = 0}^\infty
{\frac{{\left( { - 1} \right)^m }}{{m!\Gamma \left( {m + \alpha  + 1}
\right)}}} \left( {\frac{x}{2}} \right)^{2m + \alpha }
\label{eq61}
\end{equation}

One could have also replaced the eigenvalue function with the order as

\begin{equation}
x^2 \psi ''\left( x \right) + x\psi '\left( x \right) + \left[ {x^2  -
k^2 \left( x \right)} \right]\psi \left( x \right) = 0
\label{eq62}
\end{equation}

\noindent
Hence

\begin{equation}
\mathbb{H}_k  = x^2 \frac{{\partial ^2 }}{{\partial x^2 }} +
x\frac{\partial }{{\partial x}} + \left( {x^2  - k^2 } \right)\quad
\quad \mathbb{L}_{k\left( x \right)}  = x^2 \frac{{d^2 }}{{dx^2 }} +
x\frac{d}{{dx}} + \left[ {x^2  - k^2 \left( x \right)} \right]
\label{eq63}
\end{equation}

\noindent
Again since $J_\alpha  \left( x \right)$ and
$J_{ - \alpha } \left( x \right)$ for non-integer
$\alpha $ are linearly independent, the natural
choice of eigenfunctions based on Bessel functions would be either

\begin{equation}
A\left[ {x;k\left( x \right)} \right] = J_{ + k\left( x \right)}
\left( x \right)\quad B\left[ {x;k\left( x \right)} \right] = J_{ -
k\left( x \right)} \left( x \right)
\label{eq64}
\end{equation}

\noindent
or

\begin{equation}
A\left[ {x;k\left( x \right)} \right] = J_{k\left( x \right)} \left( x
\right)\quad B\left[ {x;k\left( x \right)} \right] = N_{k\left( x
\right)} \left( x \right)
\label{eq65}
\end{equation}

The latter pair has the same Wronskian of $W\left[
{x;k\left( x \right)} \right] = {2 \mathord{\left/
{\vphantom {2 {\pi x}}} \right.
\kern-\nulldelimiterspace} {\pi x}}$. Hence, the elements of
the Kernel matrix are

\begin{equation}
\begin{array}{l}
u_{11} \left( x \right) = \frac{{\pi xk'\left( x \right)}}{2}\left\{
{\frac{{\partial ^2 J_{k\left( x \right)} \left( x \right)}}{{\partial
x\partial k}}N_{k\left( x \right)} \left( x \right) - \frac{{\partial
J_{k\left( x \right)} \left( x \right)}}{{\partial k}}\frac{{\partial
N_{k\left( x \right)} \left( x \right)}}{{\partial x}}} \right\} \\
u_{12} \left( x \right) = \frac{{\pi xk'\left( x \right)}}{2}\left\{
{\frac{{\partial ^2 N_{k\left( x \right)} \left( x \right)}}{{\partial
x\partial k}}N_{k\left( x \right)} \left( x \right) - \frac{{\partial
N_{k\left( x \right)} \left( x \right)}}{{\partial k}}\frac{{\partial
N_{k\left( x \right)} \left( x \right)}}{{\partial x}}} \right\} \\
u_{21} \left( x \right) = \frac{{\pi xk'\left( x \right)}}{2}\left\{
{\frac{{\partial J_{k\left( x \right)} \left( x \right)}}{{\partial
k}}\frac{{\partial J_{k\left( x \right)} \left( x \right)}}{{\partial
x}} - \frac{{\partial ^2 J_{k\left( x \right)} \left( x
\right)}}{{\partial x\partial k}}J_{k\left( x \right)} \left( x
\right)} \right\} \\
u_{22} \left( x \right) = \frac{{\pi xk'\left( x \right)}}{2}\left\{
{\frac{{\partial J_{k\left( x \right)} \left( x \right)}}{{\partial
x}}\frac{{\partial N_{k\left( x \right)} \left( x \right)}}{{\partial
k}} - J_{k\left( x \right)} \left( x \right)\frac{{\partial ^2
N_{k\left( x \right)} \left( x \right)}}{{\partial x\partial k}}}
\right\} \\
\end{array}
\label{eq66}
\end{equation}

\subsection{Euler-Cauchy Equation}

The Euler-Cauchy equation reads

\begin{equation}
x^2 \psi ''\left( x \right) + x\psi '\left( x \right) - k^2 \left( x
\right)\psi \left( x \right) = 0
\label{eq67}
\end{equation}

\noindent Therefore, the corresponding operators are given by

\begin{equation}
\mathbb{H}_k  = x^2 \frac{{\partial ^2 }}{{\partial x^2 }} +
x\frac{\partial }{{\partial x}} - k^2 \quad \mathbb{L}_{k\left( x
\right)}  = x^2 \frac{{d^2 }}{{dx^2 }} + x\frac{d}{{dx}} - k^2 \left( x
\right)
\label{eq68}
\end{equation}

\noindent When $k\left( x \right)$ is a constant, solutions
of (\ref{eq67}) are given by linear combinations of  $x^{ + k} $
and $x^{ - k} $. Thus, the extended basis functions
are

\begin{equation}
A\left[ {x;k\left( x \right)} \right] = x^{ + k\left( x \right)} \quad
B\left[ {x;k\left( x \right)} \right] = x^{ - k\left( x \right)}
\label{eq69}
\end{equation}

\noindent
with the Wronskian $W\left[ {x;k\left( x \right)} \right] =
 - {{2k\left( x \right)} \mathord{\left/
{\vphantom {{2k\left( x \right)} x}} \right.
\kern-\nulldelimiterspace} x}$, which vanishes at
$k\left( x \right) = 0$. The Kernel matrix
simplifies to the convenient form

\begin{equation}
\left[ {U\left( x \right)} \right] =  - \frac{{k'\left( x
\right)}}{{2k\left( x \right)}}\left[ {\begin{array}{*{20}c}
{1 + 2k\left( x \right)\ln x} & { - x^{ - 2k\left( x \right)} }  \\
{ - x^{2k\left( x \right)} } & {1 - 2k\left( x \right)\ln x}  \\
\end{array}} \right]
\label{eq70}
\end{equation}

\subsection{Approximate Solution}

The exact matrix exponential in (\ref{eq23})
can be evaluated if off-diagonal elements of the Kernel matrix could
be dropped. This has been previously shown to lead to the well-known
WKB solution \cite{ref2b,ref4,ref5}. Under such conditions, the approximate differential
transfer matrix takes the form

\begin{equation}
\begin{array}{l}
\left[ {Q^{x_1  \to x_2 } } \right] = \mathbb{T}\exp \left\{
{\int\limits_{x_1 }^{x_2 } {\left[ {U\left( x \right)} \right]dx} }
\right\} \approx \mathbb{T}\exp \left\{ {\int\limits_{x_1 }^{x_2 }
{\left[ {\begin{array}{*{20}c}
{u_{11} \left( x \right)} & 0  \\
0 & {u_{22} \left( x \right)}  \\
\end{array}} \right]dx} } \right\} \\
\quad \quad \; = \left[ {\begin{array}{*{20}c}
{\exp \left[ {\int\limits_{x_1 }^{x_2 } {u_{11} \left( x \right)dx}
} \right]} & 0  \\
0 & {\exp \left[ {\int\limits_{x_1 }^{x_2 } {u_{22} \left( x
\right)dx} } \right]}  \\
\end{array}} \right] \\
\end{array}
\label{eq71}
\end{equation}

\subsection{Periodic Perturbations}

There is a great deal of simplification possible, when the eigenvalue
function $k\left( x \right) = k\left( {x + L}
\right)$ is periodic for some $L > 0$. Let
$\mathbb{T}_L $ be the translation operator, defined
as $\mathbb{T}_L h\left( x \right) = h\left( {x + L}
\right)$, for all arbitrary functions $h\left( x
\right)$. Hence, we readily have the commutation property
$\left[ {\mathbb{L}_{k\left( x \right)} ,\mathbb{T}_L } \right]
= 0$, and hence these two operators share identical
eigenfunctions. Since $\mathbb{L}_{k\left( x \right)}
$ is supposed to be linear, Bloch-Floquet theorem \cite{ref6}
applies and then any solution will take the form of Bloch
eigenfunctions

\begin{equation}
\mathbb{T}_L f\left[ {x;k\left( x \right)} \right] = f\left[ {x +
L;k\left( {x + L} \right)} \right] = \exp \left( { - j\kappa L}
\right)f\left[ {x;k\left( x \right)} \right]
\label{eq72}
\end{equation}

\noindent
in which the complex number $\kappa $ is being
referred to as the Bloch number. This shows that $f\left[
{x;k\left( x \right)} \right]$ is an eigenfunction of the
translation operator $\mathbb{T}_L $ with the
eigenvalue $\exp \left( { - j\kappa L} \right)$.
Based on (\ref{eq72}), we furthermore have

\begin{equation}
f\left( x \right) = \exp \left( { - j\kappa x} \right)g_\kappa  \left(
x \right)
\label{eq73}
\end{equation}

\noindent
with the envelope function satisfying

\begin{equation}
\mathbb{T}_L g_\kappa  \left( x \right) = g_\kappa  \left( x \right)
\label{eq74}
\end{equation}

Under such circumstances, it is easy to obtain the characteristic
equation of eigenfunctions. From (\ref{eq5}) and (\ref{eq72})
we have

\begin{equation}
\begin{array}{l}
\mathbb{T}_L f\left[ {x;k\left( x \right)} \right] = a\left( {x + L}
\right)A\left[ {x + L;k\left( x \right)} \right] + b\left( {x + L}
\right)B\left[ {x + L;k\left( x \right)} \right] \\
\quad \quad \quad \;\;\, = \exp \left( { - j\kappa L} \right)\left\{
{a\left( x \right)A\left[ {x;k\left( x \right)} \right] + b\left( x
\right)B\left[ {x;k\left( x \right)} \right]} \right\} \\
\end{array}
\label{eq75}
\end{equation}

\noindent
By taking differentiating with respect to $x$ from
(\ref{eq72}) and derivative lemma (\ref{eq30})
we also get

\begin{equation}
\begin{array}{l}
\mathbb{T}_L f'\left[ {x;k\left( x \right)} \right] = \exp \left( { -
j\kappa L} \right)f'\left[ {x;k\left( x \right)} \right] \\
\quad \quad \quad \;\;\; = \exp \left( { - j\kappa L} \right)\left\{
{a\left( x \right)A_x \left[ {x;k\left( x \right)} \right] + b\left( x
\right)B_x \left[ {x;k\left( x \right)} \right]} \right\} \\
\end{array}
\label{eq76}
\end{equation}

We furthermore notice that (\ref{eq75})
and (\ref{eq76}) are actually identities which hold for all $x$.
These two can be combined to get the system of equations

\begin{equation}
\begin{array}{l}
\left[ {\begin{array}{*{20}c}
{A\left[ {x + L;k\left( x \right)} \right]} & {B\left[ {x +
L;k\left( x \right)} \right]}  \\
{A_x \left[ {x + L;k\left( x \right)} \right]} & {B_x \left[ {x +
L;k\left( x \right)} \right]}  \\
\end{array}} \right]\left\{ {C\left( {x + L} \right)} \right\} =  \\
\quad \quad \exp \left( { - j\kappa L} \right)\left[
{\begin{array}{*{20}c}
{A\left[ {x;k\left( x \right)} \right]} & {B\left[ {x;k\left( x
\right)} \right]}  \\
{A_x \left[ {x;k\left( x \right)} \right]} & {B_x \left[ {x;k\left(
x \right)} \right]}  \\
\end{array}} \right]\left\{ {C\left( x \right)} \right\} \\
\end{array}
\label{eq77}
\end{equation}

\noindent
which allows the solution

\begin{equation}
\begin{array}{l}
\mathbb{T}_L \left\{ {C\left( x \right)} \right\} = \frac{{\exp \left(
{ - j\kappa L} \right)}}{{W\left[ {x + L;k\left( x \right)}
\right]}}\left[ {\begin{array}{*{20}c}
{B_x \left[ {x + L;k\left( x \right)} \right]} & { - B\left[ {x +
L;k\left( x \right)} \right]}  \\
{ - A_x \left[ {x + L;k\left( x \right)} \right]} & {A\left[ {x +
L;k\left( x \right)} \right]}  \\
\end{array}} \right] \times  \\
\quad \quad \quad \quad \quad \quad \quad \quad \quad \left[
{\begin{array}{*{20}c}
{A\left[ {x;k\left( x \right)} \right]} & {B\left[ {x;k\left( x
\right)} \right]}  \\
{A_x \left[ {x;k\left( x \right)} \right]} & {B_x \left[ {x;k\left(
x \right)} \right]}  \\
\end{array}} \right]\left\{ {C\left( x \right)} \right\} \\
\quad \quad \quad \; = \exp \left( { - j\kappa L} \right)\left[ V
\right]\left\{ {C\left( x \right)} \right\} \\
\end{array}
\label{eq78}
\end{equation}

\noindent But from (\ref{eq26})

\begin{equation}
\mathbb{T}_L \left\{ {C\left( x \right)} \right\} = \left[ {Q^{x \to x +
L} } \right]\left\{ {C\left( x \right)} \right\}
\label{eq79}
\end{equation}

\noindent
Simultaneous satisfaction of (\ref{eq77}) and (\ref{eq78})
requires that

\begin{equation}
\left| {\exp \left( { - j\kappa L} \right)\left[ V \right] - \left[ Q
\right]} \right| = 0
\label{eq80}
\end{equation}

\noindent
In other words, we should have

\begin{equation}
\kappa  = \frac{j}{L}\ln \left\{ {{\mathop{\rm eig}\nolimits} \left[ P
\right]} \right\}
\label{eq81}
\end{equation}

\noindent
in which $\left[ P \right] = \left[ V \right]^{ - 1} \left[
Q \right]$. Here, constancy of the Bloch wavenumber
$\kappa $ is guaranteed by Bloch-Floquet theorem,
and is rigorously shown to hold for the example of extended wave
equation (\ref{eq49}) with periodic wavenumber in \cite{ref6}.

\section{Conclusions}

In this paper, we presented a new analytical solution obtained by
differential transfer matrix method to a wide class of second-order
linear differential equations, which are extended from eigenvalue
problems by replacing the eigenvalue with an arbitrary eigenvalue
function. We presented the details of the method and a fundamental
theorem to rigorously establish the mathematical formulation. Few
examples were also described.

\end{document}